\documentclass[
final
]{dmtcs-episciences}

\publicationdata{vol. 26:3}{2024}{23}{10.46298/dmtcs.11175}{2023-04-11; 2023-04-11; 2024-06-04}{2024-12-12}

\usepackage[utf8]{inputenc}
\usepackage{subfigure}

\usepackage[round]{natbib}

\usepackage{amsthm}

\usepackage[boxed,lined,noend,algoruled]{algorithm2e}
\usepackage{amsmath,amssymb}
\usepackage{tikz}
\usepackage{tkz-euclide}
\usepackage{url}
\usepackage{pgfplots}
\usepackage{xcolor}
\usepackage{xspace}
\usepackage{framed}
\usepackage[T1]{fontenc}

\DeclareSymbolFont{largesymbolsA}{U}{txexa}{m}{n}
\DeclareMathSymbol{\varprod}{\mathop}{largesymbolsA}{16}

\DeclareMathAlphabet\mathbfcal{OMS}{cmsy}{b}{n}

\newcommand{\itpaths}[1]{\G^{#1\rightarrow t}}
\newcommand{\itpathsless}[2]{\G^{#1\rightarrow t}_{#2}}

\newcommand{\Metric}{\mathcal{M}}

\newcommand{\G}{\mathcal{G}}
\newcommand{\GG}{\mathbb{G}}
\newcommand{\VV}{\mathbb{V}}
\newcommand{\EE}{\mathbb{E}}

\newcommand{\R}{\mathbb{R}}

\newcommand{\U}{\mathcal{U}}
\renewcommand{\P}{\mathcal{P}}

\newcommand{\UU}{\mathcal{U}}

\newcommand{\NP}{\mathcal{NP}}

\newcommand{\grandO}{\mathcal{O}}
\newcommand{\uu}{u}

\newcommand{\ddet}{d}
\newcommand{\drob}{d}
\newcommand{\nU}[1]{{\sigma_{#1}}}
\newcommand{\NU}{\sigma}

\newcommand{\dmax}{d^{max}}

\newcommand{\costmax}{\cost_{d^{max}}}
\newcommand{\costmaxM}[1]{\cost^{max}}

\newcommand{\cost}{c}
\newcommand{\costtilde}{\tilde{c}}
\newcommand{\Ex}{\mathbb{E}}
\newcommand{\abs}[1]{\left\lvert #1\right\rvert}

\newcommand{\floor}[1]{\left\lfloor#1\right\rfloor}

\newcommand{\distany}[2]{d\left(#1,#2\right)}
\newcommand{\dist}[2]{\distany{#1}{#2}}

\newcommand{\distbisMM}[1]{\tilde{d}}

\DeclareMathOperator*{\diam}{diam}
\DeclareMathOperator*{\ext}{ext}
\newtheorem{observation}{Observation}

\newcommand{\revision}[1]{{#1}}

\newcommand{\SP}{\textsc{gen-sp}\xspace}

\newcommand{\MINEWCP}{\textsc{min-ewcp}\xspace}
\newcommand{\DELTASP}{{\textsc{sp}}\xspace}
\newcommand{\MSF}{{\textsc{min-ssf}}\xspace}
\newcommand{\DELTAMST}{{\textsc{mst}}\xspace}
\newcommand{\DELTATSP}{{\textsc{tsp}}\xspace}
\newcommand{\PBFAMILY}{\ensuremath{\mathcal{S}}\xspace}

\newcommand{\ROBUST}[1]{\ensuremath{\textsc{robust}\mbox{-#1}}\xspace}

\newcommand{\additionalInput}{\alpha}
\newcommand{\OPT}{\ensuremath{\textsc{opt}}\xspace}

\newcommand{\NPH}{$\cal{NP}$-hard\xspace}

\newcommand{\PTAS}{$\cal{PTAS}$\xspace}

\newcommand{\FPTAS}{$\cal{FPTAS}$\xspace}

\DeclareMathOperator*{\Val}{Val}
\newcommand{\nval}{n_{val}}
\newcommand{\nprof}{n_{\P}}
\renewcommand{\P}{\mathcal{P}} 

\DeclareMathOperator*{\argmin}{arg\,min}
\DeclareMathOperator*{\argmax}{arg\,max}

\newtheorem{theorem}{Theorem}

\newtheorem{corollary}{Corollary}

\newtheorem{example}{Example}

\newtheorem{lemma}{Lemma}

\newtheorem{proposition}{Proposition}
\newtheorem{remark}{Remark}

\usepackage{hyperref}
\pgfplotsset{compat=1.18}

\author[Bougeret et al.]{Marin Bougeret\affiliationmark{1}
  \and J\'er\'emy Omer\affiliationmark{2}
  \and  Michael Poss\affiliationmark{1}}
\title[Optimization with locational uncertainty]{Approximating optimization problems in graphs with locational uncertainty}
\affiliation{
  LIRMM, University of Montpellier, CNRS, Montpellier, France\\
  IRMAR, INSA Rennes, Rennes, France}

\keywords{robust optimization, approximation algorithms, dynamic programming}

\begin{document}

\maketitle

\begin{abstract}
  Many combinatorial optimization problems can be formulated as the search for a subgraph that satisfies certain properties and minimizes the total weight. We assume here that the vertices correspond to points in a metric space and can take any position in given uncertainty sets. Then, the cost function to be minimized is the sum of the distances for the worst positions of the vertices in their uncertainty sets. We propose two types of polynomial-time approximation algorithms. The first one relies on solving a deterministic counterpart of the problem where the uncertain distances are replaced with maximum pairwise  distances. We study in details the resulting approximation ratio, which depends on the structure of the feasible subgraphs and whether the metric space is Ptolemaic or not. The second algorithm is a fully-polynomial time approximation scheme for the special case of $s-t$ paths.
\end{abstract}

\section{Introduction}

Given a graph $\GG=(\VV,\EE)$ and a family of feasible subgraphs $\G$ of $\GG$, many discrete optimization problems amount to find the cheapest subgraph in $\G$. In this paper, we assume that $\GG$ is a graph embedded into a given metric space $(\Metric,d)$ so that every vertex $i\in\VV$ is associated to a point $u_i\in \Metric$ and the weight of edge $\{i,j\}\in \EE$ is equal to the distance $d(u_i,u_j)$. Then, the cost of a subgraph $G\in\G$ is defined as the sum of the weights of its edges. Denoting by $E[G]$ the set of edges of $G$, the cost of graph $G\in\G$ is given by 
$$
c_d(\uu,G)=\sum_{\{i,j\}\in E[G]} d(u_i,u_j),
$$
leading to the combinatorial optimization problem
\begin{equation}
\tag{$\Pi$}
\label{eq:CO}
 \min_{G\in \G}\sum_{\{i,j\}\in E[G]} d(u_i,u_j).
\end{equation}
\revision{The subscript $\drob$ on $c_d$ will be removed when clear from the context, denoting the cost by $c(u,G)$.}

An example of problem~\ref{eq:CO} arises in data clustering, where one wishes to partition a given set of data points $\{u_1,\ldots,u_n\}\subset\Metric$ into at most $K$ sets so as to minimize the sum of all dissimilarities between points that belong to the same element of the partition. In this example, $\GG$ is a complete graph, every vertex $i$ of $\VV$ represents one of the data point $u_i$, and $\G$ consists of all disjoint unions of cliques covering $\GG$. Furthermore, the distance $d(u_i,u_j)$ between any two vertices $i,j\in \VV$ measures the dissimilarity between the corresponding data points. This dissimilarity may involve Euclidean distances, for instance when comparing numerical values, as well as more ad-hoc distances when comparing ordinal values instead. Problem~\ref{eq:CO} encompasses other applications, such as subway network design or facility location. In the the network design problem, one wishes to find a cheapest feasible subgraph of $\GG$ connecting a certain number of vertices. Thus, $\G$ consists of Steiner trees covering a given set of terminals and the distance used is typically Euclidean, see~\cite{GUTIERREZJARPA20133000}. In facility location problems, set $\G$ consists of unions of disjoint stars, the centers and the leaves of the stars respectively representing the facility and the clients. In that example, the distance typically involves shortest paths in an underlying road network represented by an auxiliary weighted graph $G_\Metric=(V_\Metric,E_\Metric)$, leading to a graph-induced metric space, see~\cite{MELKOTE2001515}.

More formally, we are interested in the class $\PBFAMILY$ of deterministic combinatorial optimization problems that can be formulated as~\ref{eq:CO}. 
Any $\Pi\in \PBFAMILY$ represents a specific problem, such as the shortest path or the minimum spanning tree.
An instance of problem $\Pi$ is defined by its input $I=(\GG,\additionalInput,\uu,\ddet)$, where $\GG$ is an undirected simple graph, $\additionalInput$ contains problem-specific additional inputs, and \revision{$\ddet$ is the distance matrix of the set $\bigcup_{i\in \VV}\{u_i\}$.} We denote by $\G(I)$ the set of all subgraphs (not necessarily induced) of $\GG$ that satisfy the constraints specific to $\Pi$ for input $I$. 
For instance, for the Shortest Path problem (\DELTASP), the additional input $\additionalInput=\{s,t\}$ consists of the origin and destination vertices, while $\additionalInput=\emptyset$ 
in the Minimum Spanning Tree problem (\DELTAMST). 
Notice that considering the distance $d$ as a matrix and providing it explicitly in the input avoids (i) providing $(\Metric,d)$ as part of the input, which may not be straightforward in some cases (e.g., if $(\Metric,d)$ is a Riemannian manifold), (ii) discussing the complexity of computing the distances which again might involve solving difficult optimization problems.

It has been assumed so far that the position of the vertices are known with precision, which is often not a realistic assumption. In data clustering, uncertainty on the values is common, be it because of a measurement error, or because of a lack of information -- in which case the corresponding coordinate is often replaced by the full interval~(\cite{masson2020cautious}). Similar uncertainties arise in network design or facility location because the exact location of the stations and facilities must satisfy technical requirements as well as political considerations as local officials are never happy to let their citizens face the inconvenience of heavy civil engineering. In this paper, we address this issue through the lens of robust optimization, \revision{by introducing $\ROBUST{$\Pi$}$. An instance of \ROBUST{$\Pi$} is given by $I=(\GG,\additionalInput,\U,\drob)$ where $\GG$ and $\alpha$ are as before, $\drob$ is a distance matrix, $\U_i$ denotes the set of possible positions for node $i$, which correspond to indexes of rows/columns of the the distance matrix $\drob$, and $\UU=\times_{i\in \VV}\U_i$. Given this instance, the problem is to solve
\begin{equation}
\tag{\ROBUST{$\Pi$}}
\label{eq:minmaxCO}
\min_{G\in \G(I)} \max_{\uu\in\UU}\sum_{\{i,j\}\in E[G]} d(u_i,u_j).
\end{equation}
Thus, the objective of \ROBUST{$\Pi$} is to find $G\in\G(I)$ that minimizes $c(G) = \max_{\uu\in\U} c(\uu,G)$. Introducing the notation $\nU{i}=|\U_i|$, we see that \ROBUST{$\Pi$} is a generalization of $\Pi$ as it corresponds to the case where $\nU{i}=1$ for each $i$.

\begin{remark}
When $(\Metric,d)$ is the $q$-dimensional Euclidean space, a natural setting for locational uncertainty would be to model the uncertainty around $u_i$ by a convex set, such as a polytope $\P_i\subset\R^q$, with $\P=\times_{i\in \VV}\P_i$. Then, given the set of subgraphs $\G(I)$ of $\GG$, one could seek to solve the optimization problem
\begin{equation}
\label{eq:generalrobust}
\min_{G\in \G(I)} \max_{\uu\in\P}\sum_{\{i,j\}\in E[G]} d(u_i,u_j),
\end{equation}
where the adversary maximizes over polytope $\P$ instead of the finite set $\U$ used in~\eqref{eq:minmaxCO}. Nevertheless, due to the convexity of $d$, we have $$\max_{\uu\in\P}\sum_{\{i,j\}\in E[G]} d(u_i,u_j)= \max_{\uu\in\ext(\P)}\sum_{\{i,j\}\in E[G]} d(u_i,u_j),$$ where $\ext(\P)$ denotes the set of extreme points of $\P$. Therefore, problem~\eqref{eq:generalrobust} can be cast into our setting by letting $\drob$ be the distance matrix of the set $\bigcup_{i\in\VV}\ext(\P_i)$ and defining the sets $\U_i$ accordingly.
\end{remark}
}

Problem~\ref{eq:minmaxCO} is closely related to min-max robust combinatorial optimization pioneered in~\cite{kouvelis2013robust}. In that framework, one searches for the best solution to a combinatorial optimization problem given that the adversary chooses the worst possible cost vector in a given uncertainty set, see the surveys by~\cite{AissiBV09,BuchheimK18,kasperski2016robust}. The complexity and approximation ratios available for the resulting problems typically depend on the underlying combinatorial optimization problem (herein represented by~\ref{eq:CO}) and the structure of the uncertainty set. While some algorithms have been proposed to handle arbitrary finite uncertainty sets (e.g.~\cite{ChasseinGKZ20,KasperskiZ13}), simpler sets typically benefit from stronger results, such as axis-parallel ellipsoids~(\cite{BaumannBI14,Nikolova2010}) or budget uncertainty sets~(\cite{BertsimasS03,BertsimasS04}). The latter results have been extended to more general polytopes in~\cite{OJMO_2024__5__A5_0} and to problems featuring integer decision variables in~\cite{GoetzmannST11}. Specific combinatorial optimization problems, such as the shortest path or the spanning tree, have also benefited from stronger results, e.g.~\cite{AissiBV05,YamanKP01}.

An important specificity of problem~\ref{eq:minmaxCO} is that its cost function is in general non-concave, hardening the inner maximization problem. 
In fact, we have shown in our companion paper~(\cite{BougeretOmerPoss1}) that computing the cost function $c(G)$ of a given subgraph $G\in \G(I)$ is in general $\NP$-hard. 
We have further shown in that paper that problem~\ref{eq:minmaxCO} is $\NP$-hard for the shortest path and the minimum spanning tree. 
That work also proposes exact and approximate solution algorithms based on mixed-integer formulations. 
Another paper addressing~\ref{eq:minmaxCO} is~\cite{citovsky2017tsp}, which relies on computational geometry techniques to provide constant-factor approximation algorithms in the special case where \ref{eq:CO} is the traveling salesman problem, meaning that $\G(I)$ consists of all Hamiltonian cycles of $\GG$. The main result of that paper is a polynomial-time approximation scheme for the special case in which each $\U_i$ is a set of disjoint unit disks in the plane. They also propose an approximation algorithm that amounts to solve a deterministic counterpart of~\ref{eq:minmaxCO} where the uncertain distances are replaced by the maximum pairwise distances 
$$
\dmax_{ij}=\max_{u_i\in \U_i, u_j\in \U_j} \dist{u_i}{u_j},
$$
for each $(i,j)\in \VV^2, i\neq j$.

In Section~\ref{sec:approx}, we deal with the main purpose of this paper, which is to extend the approximation algorithm based on $\dmax$ and suggested by~\cite{citovsky2017tsp} to general sets $\G(I)$ and distances matrices more general than those induced by embeddings into Euclidean spaces. 
First, Theorem~\ref{thm:approx} transfers any approximation ratio known for problem $\Pi\in\PBFAMILY$ to \ROBUST{$\Pi$}. We prove that, in general, the transfer involves a multiplicative constant equal to 9 (Theorem~\ref{thm:constant_ratio_non_pto}). 
Then, we dig into smaller multiplicative constants, assuming that $\G(I)$ and/or $d$ satisfy additional assumptions. Regarding $\G(I)$, we obtain smaller constants for special families of graphs such as cycles, stars, trees, and graphs that can be composed as disjoint unions of these graphs. As a special case of our results, we find the constant of 3 for Hamiltonian cycles previously provided by~\cite{citovsky2017tsp}. Concerning the structure of $d$, we show that distances that satisfy Ptolemy's inequality (see~\cite{apostolPtolemyInequalityChordal1967}) benefit from stronger results.

We complement these results by Section~\ref{sec:FPTAS} where we provide a dynamic programming algorithm for the special case where \ref{eq:CO} is the Shortest Path problem, so $\G(I)$ consists of all $s-t$ paths. The algorithm is then turned into a fully-polynomial time approximation scheme by rounding the input appropriately.

\paragraph{Additional notations and definitions}

Given a simple undirected graph $G$, $V[G]$ denotes its set of vertices.
When clear from context we use notations $n=|V[G]|$ and $m=|E[G]|$, where $|S|$ denotes the cardinality of any finite set $S$.
For any $i \in V[G]$, we denote by $N(i)=\{j \in V[G] \mid \{i,j\} \in E[G] \}$ the neighborhood of $i$. We say that a graph $G$ is a clique if for any two disjoint vertices $i,j$,$\{i,j\} \in E$, and we say that \revision{$G$ is a star if $G$ is a complete bipartite
graph $K_{1,k}$ for some $k \ge 1$}. For any positive integer $k$, we denote $[k]=\{1,\dots,k\}$. The diameter of a subet $S$ of the indices of $d$ is given by $\max_{u,v\in S}d(u,v)$. \revision{Remember also that given a problem \ROBUST{$\Pi$}, $\GG$ denotes the graph in the input $I$,  $\G(I)$ is the set of all subgraphs (not necessarily induced) $G$ of $\GG$ that satisfy the constraints specific to $\Pi$ for input $I$.}
\section{Approximation of the general robust problem}
\label{sec:approx}

\subsection{Reduction to a deterministic problem by using worst-case distances}
\label{sec:wcd}

\begin{algorithm}[!ht]
\DontPrintSemicolon
\SetKwInOut{Return}{return}
\SetKwInOut{Initialization}{initialization}
\revision{
Given an instance $I=(\GG,\additionalInput,\U,\drob)$ of \ROBUST{$\Pi$}\;
Select $u' \in \UU$ to define instance $I'=(\GG,\additionalInput,u',\drob')$ of $\Pi$, where $\drob'$ is the submatrix of $\drob$ corresponding to $u'$\;
Compute $G$ using an approximation algorithm for $I'$\;
\textbf{return} $G$
}
\caption{Solving a deterministic counterpart based on some representative location~$\uu'$}
\label{algo:approx1}
\end{algorithm}

A natural approach to \ROBUST{$\Pi$} would be to choose a relevant vector $\uu'\in\U$, reduce to the corresponding deterministic problem $\Pi$ given by
$
\min_{G\in \G(I)}\sum_{\{i,j\}\in E[G]} \dist{u'_i}{u'_j},
$
and use any known approximation algorithm for $\Pi$. This approach is formalized by Algorithm~\ref{algo:approx1}. 

Unfortunately, choosing such a representative $\uu'$ is not easy. For instance, a natural choice might be to consider the geometric median of each set, i.e., $\uu'=\uu^{gm}$ where 
$
 u_i^{gm} \in \argmin_{u_1\in \U_i} \sum_{u_2\in\U_i}\dist{u_1}{u_2}.
$
Although the choice of geometric median may appear natural at first glance, the cost of the solution $G^{gm}$ returned by Algorithm~\ref{algo:approx1} for $\uu'=\uu^{gm}$ may actually be arbitrarily larger than the optimal solution cost.
\begin{observation}
  Let $\Pi\in\PBFAMILY$ such that \revision{any $G\in\G(I)$ is a single edge 
  for any instance $I$ of $\Pi$}. Let $e^{gm}$ be the solution returned by 
  Algorithm~\ref{algo:approx1} for $\uu'=\uu^{gm}$. The ratio $\cost(e^{gm})/\OPT$ is unbounded.
\end{observation}
\begin{proof}
 For any $\epsilon>0$ small enough, consider $\VV=\{1,2,3\}$ and $\EE=\{\{1,2\},\{2,3\}\}$. We consider an embedding into the $1$-dimensional Euclidean space with $\U_1=\{\epsilon\}$, $\U_2=\{0\}$, and $\U_3=\{-1,0,1\}$, and defining $d$ as the resulting distance matrix. We have that $u_1^{gm}=\epsilon$ and $u_2^{gm} = u_3^{gm} = 0$ so Algorithm~\ref{algo:approx1} picks edge $\{2,3\}$ having a cost of $\cost(\{2,3\})= 1$. In contrast edge $\{1,2\}$ has a cost of $\epsilon$, yielding a ratio of $1/\epsilon$.
\end{proof}

We proceed by using a different approach for reducing to the deterministic problem~$\Pi{}$. 
Given an instance $I=(\GG,\alpha,\U,\drob)$ of \ROBUST{$\Pi$}, we construct an instance $I'=(\GG,\alpha,v,\dmax)$ of $\Pi$ by defining $\dmax_{ij}=\max_{u_i\in \U_i, u_j\in \U_j} \dist{u_i}{u_j}$ and $v=(1,\ldots,|\VV|)$. Notice that $\dmax$ is indeed a distance matrix as $\dmax_{ij}=0$ iff $i=j$ and it satisfies the triangular inequality. \revision{This leads to Algorithm~\ref{algo:approx}.}

\begin{algorithm}[!ht]
\DontPrintSemicolon
\SetKwInOut{Return}{return}
\SetKwInOut{Initialization}{initialization}
\revision{
Given an instance $I=(\GG,\additionalInput,\U,\drob)$ of \ROBUST{$\Pi$}\;
Construct $\dmax$ to define instance $I'=(\GG,\additionalInput,v,\dmax)$ of $\Pi$\;
Compute $G$ using an approximation algorithm for $I'$\;
\textbf{return} $G$
}
\caption{Solving a deterministic counterpart based on $\dmax$ distances}
\label{algo:approx}
\end{algorithm}

\revision{
\begin{theorem}
  \label{thm:approx}
  Let $\Pi \in\PBFAMILY$ and assume that:
  \begin{itemize}
    \item $\Pi$ has a $\rho_1$-polynomial time approximation algorithm, and
    \item there exists $\rho_2 \ge 1$ such that for each input $I=(\GG,\additionalInput,\U,\drob)$ of \ROBUST{$\Pi$}, $\cost_{\dmax}(G)\leq \rho_2 \cost_\drob(G)$ for any $G\in\G(I)$.
  \end{itemize}
Then, Algorithm~\ref{algo:approx} can be used to derive a polynomial $\rho_1\rho_2$-approximation for \ROBUST{$\Pi$}.

\end{theorem}


\begin{proof}
  Let $G^*$ be an optimal solution of instance $(\GG,\alpha,\U,\ddet)$ of \ROBUST{$\Pi$}, and $G^{max}$ be an optimal solution to instance $I'=(\GG,\alpha,v,\dmax)$ of $\Pi$. Furthermore, following the first assumption, there is a $\rho_1$-approximation algorithm for $\Pi$, which we use to construct a solution $G$ to $I'$ such that $\costmax(G) \le \rho_1 \costmax(G^{max})$.   
  We have $$\cost_d(G) \le \cost_{\dmax}(G) \le \rho_1 \cost_{\dmax}(G^{max}) \le \rho_1\cost_{\dmax}(G^*) \le \rho_1\rho_2 \cost_d(G^*),$$ where the last inequality follows from the second assumption of the theorem.
\end{proof}

}

\begin{table}[!h]
\footnotesize
  \begin{tabular}{p{0.22\textwidth}||p{0.20\textwidth}|p{0.25\textwidth}|p{0.26\textwidth}}
    \hline
Problem (each connected component belongs to graph family $\G(I)$) & deterministic \newline version & $\cost(G)/\costmax(G)$ & robust counterpart \\
\hline
\hline
$\Pi \in \PBFAMILY$ (arbitrary) & $\rho$-approx. & $9$ (Thm~\ref{thm:constant_ratio_non_pto}) \newline $4$ (Ptolemaic, Thm~\ref{thm:constant_ratio_general})  & $9\rho$-approx. \newline $4\rho$-approx. (Ptolemaic)  \\
\hline
\MINEWCP (clique) & $2$-approx~(\cite{eremin20142}) & $2$ (Prop~\ref{prop:cliques}) & $4$-approx. \\
\hline
\DELTASP (path) & polynomial & $2$ (Cor~\ref{cor:ratio_paths_cycles}) &  $2$ approx \newline \NPH (\cite{BougeretOmerPoss1}) \newline  \FPTAS (Thm~\ref{thm:fptas}) \\
\hline
\DELTAMST (tree) & polynomial & $6$ (Prop~\ref{prop:tree}) \newline $4$ (Ptolemaic, Thm~\ref{thm:constant_ratio_general}) \newline $2\sqrt 2$ (planar Eucl., Prop~\ref{prop:tree_eucl}) & $6$-approx. \newline $4$-approx. (Ptolemaic) \newline $2\sqrt 2$-approx. (planar Eucl.) \newline \NPH (\cite{BougeretOmerPoss1})\\
\hline
\MSF (star) & polynomial~(\cite{KHOSHKHAH20191}) & 3 (Prop~\ref{prop:stars}) \newline 2 (Ptolemaic, Cor~\ref{cor:star}) & 3-approx \newline 2-approx \\
\hline
\DELTATSP (cycle) & $\frac{3}{2}$-approx~(\cite{christofides1976worst}) & $2$ (\cite{citovsky2017tsp}) and~Cor~(\ref{cor:ratio_paths_cycles}) & $3$-approx~(\cite{citovsky2017tsp}) \newline \PTAS (planar Eucl.)~(\cite{citovsky2017tsp}) \\
\hline
  \end{tabular}
\bigskip
\caption{Overview of our results. When no reference is given, the ratios in the ``robust counterpart'' column are obtained by Theorem~\ref{thm:approx}. All our results are for finite $\U_i$, except the $2 \sqrt 2$ ratio for trees, which holds for balls in the Euclidean plane, and the \PTAS for \DELTATSP holding for disjoint unit balls in the Euclidean plane.}
\label{fig:array}
\label{tab:overview}
\end{table}

Let us recall Ptolemy's inequality (see e.g. \cite{apostolPtolemyInequalityChordal1967}).
A distance is \emph{Ptolemaic} if for any four indexes $A, B, C, D$ of $d$,
\begin{equation}
\label{eq:PtolemyIneq}
 \dist{A}{C}\cdot \dist{B}{D} \leq \dist{A}{B}\cdot \dist{C}{D} + \dist{B}{C}\cdot \dist{A}{D}.
\end{equation}
In Section~\ref{sec:generalgraphs} and Section~\ref{sec:specificgraphs} we prove the existence of constant upper bounds on the ratio $\costmax(G)/\cost(G)$ for different families~$\G(I)$ and distances. A summary of our results is given in Table~\ref{tab:overview}. 
In particular, for any graph $G$, Theorem~\ref{thm:constant_ratio_general} states that $\costmax(G)\leq 4 \cost(G)$ for any Ptolemaic distance, and Theorem~\ref{thm:constant_ratio_non_pto} states that $\costmax(G)\leq 9 \cost(G)$ for any distance.
These imply that, up to a constant factor, \ROBUST{$\Pi$} is not harder to approximate than $\Pi$, as formalized below.
\revision{
\begin{theorem}
  \label{thm:approx_bis}
Let $\Pi \in\PBFAMILY$ be a problem that is $\rho_1$-approximable in polynomial time and $I=(\GG,\additionalInput,\U,\drob)$ be an instance of \ROBUST{$\Pi$}. Then, Algorithm~\ref{algo:approx} yields a polynomial $9\rho_1$-approximation for \ROBUST{$\Pi$}, and even a $4\rho_1$-approximation if $d$ is Ptolemaic.
\end{theorem}
}

\subsection{Bounding the approximation ratio on general graphs}

\label{sec:generalgraphs}

\revision{We divide our study of $\costmax(G)/\cost(G)$ for arbitrary graphs $G$ into two cases. First, we consider that the distances matrices satisfy the Ptolemy's inequality recalled in~\eqref{eq:PtolemyIneq}. Then, we consider arbitrary distances matrices.}

\subsubsection{Ptolemaic distances}
\label{sec:ptolemaic}

\revision{We consider throughout the section that $d$ is Ptolemaic, which includes, for instance, distances resulting from embeddings into Euclidean spaces.} A direct consequence of the definition is given by the following lemma.
\begin{lemma}\label{lem:triangle}
Let $d$ be a Ptolemaic distance matrix and let $A,B$ and $C$ be such that $$\dist{B}{C}\geq \max\{\dist{A}{B},\dist{A}{C}\}.$$ Then, for any other $O$,
\begin{equation}
\label{eq:ptolemaictriangle}
d(O,A) \leq d(O,B) + d(O,C).
\end{equation}
\end{lemma}
\begin{proof}
Using $\dist{B}{C}\geq \max\{\dist{A}{B},\dist{A}{C}\}$ and Ptolemy's inequality, we get 
$$
\dist{B}{C}\cdot \dist{O}{A} \leq \dist{B}{C}\cdot \dist{O}{B} + \dist{B}{C}\cdot \dist{O}{C},
$$
and the result follows.
\end{proof}

Using the above inequality, we can get a constant bound on the approximation ratio by focusing on the extremities of a diameter of each uncertainty set $\U_i, i\in \VV$. It results in the following ratio.

\begin{theorem}\label{thm:constant_ratio_general}
Let $d$ be a Ptolemaic distance and $G$ be a graph. Then, $\costmax(G)\leq 4\:\cost_d(G)$.
\end{theorem}
\begin{proof}
 For all $i\in\{1,\dots,n\},$ let $[u_i^1,u_i^2]$ be a diameter of $\U_i,$ i.e., $u_i^1\in\U_i, u_i^2\in\U_i$ and $\dist{u_i^1}{u_i^2}=\diam(\U_i)$.
Let $\widetilde{\U}_i=\{u_i^1,u_i^2\}$ and $\widetilde{\UU}$ be the cross product of the $\widetilde{\U}_i$.
Let $\costtilde(G) = \max_{\uu\in\widetilde{\UU}}\sum_{\{i,j\}\in E[G]} d(u_i,u_j)$. 
As $\widetilde{\UU} \subseteq \UU$, we get $\cost(G) \ge \costtilde(G)$.
Let us now prove that $\costtilde(G) \ge \frac{\costmax(G)}{4}$.

Define the random variable $\tilde{u}_i$ taking any value of $\widetilde{\U}_i$ with equal probability $1/2$.
The worst-case length of the graph is not smaller than its expected edge length, i.e.,
$$\costtilde(G)=\max_{\uu\in\widetilde{\UU}}\sum_{\{i,j\}\in E[G]} \dist{u_i}{u_j}\geq \Ex\left[\sum_{\{i,j\}\in E[G]} \dist{\tilde{u}_i}{\tilde{u}_j}\right],$$
where, by linearity of expectation,
$$\Ex\left[\sum_{\{i,j\}\in E[G]} \dist{\tilde{u}_i}{\tilde{u}_j}\right] = \sum_{\{i,j\}\in E[G]} \Ex \left[\dist{\tilde{u}_i}{\tilde{u}_j}\right].$$
We then consider some arbitrary edge $\{i,j\}\in E[G]$:
\begin{align*} \Ex \left[\dist{\tilde{u}_i}{\tilde{u}_j}\right]=&\frac{1}{4}\left(\dist{u_i^1}{u_j^1} + \dist{u_i^1}{u_j^2} + \dist{u_i^2}{u_j^1} + \dist{u_i^2}{u_j^2} \right).
\end{align*}
Let $\bar{u}_i\in\U_i$ and $\bar{u}_j\in\U_j$ such that $\dist{\bar{u}_i}{\bar{u}_j}= \dmax_{ij}$.
As $u_j^1u_j^2$ is a diameter of $\U_j$, we have $d(u_j^1,u_j^2) \ge \max(d(u_j^1,\bar{u}_j),d(u_j^2,\bar{u}_j))$, and we can apply Lemma~\ref{lem:triangle} twice to the triplet $(u_j^1,u_j^2,\bar{u}_j)$ to get
\begin{equation*}
\left\{\begin{aligned}
 \dist{u_i^1}{u_j^1} + \dist{u_i^1}{u_j^2} &\geq \dist{u_i^1}{\bar{u}_j}\\
 \dist{u_i^2}{u_j^1} + \dist{u_i^2}{u_j^2} &\geq \dist{u_i^2}{\bar{u}_j}.
\end{aligned}\right.
\end{equation*}
One last application of the lemma in $u_i^1u_i^2\bar{u}_i$ then yields
$$\dist{u_i^1}{\bar{u}_j}+\dist{u_i^2}{\bar{u}_j}\geq \dist{\bar{u}_i}{\bar{u}_j}=\dmax_{ij}.$$
Summarizing the above, we get 
to $\cost(G)\ge\costtilde(G) \ge \sum_{\{i,j\}\in E[G]} \Ex \left[\dist{\tilde{u}_i}{\tilde{u}_j}\right] \ge \frac{1}{4}\costmax(G)$.
\end{proof}

\subsubsection{Arbitrary distances}
\label{sec:nonptolemaic}

Lemma~\ref{lem:triangle} does not apply to non-Ptolemaic distances, as illustrated in the following example.

\begin{figure}[ht]
\begin{center}
  \includegraphics[width=0.2\textwidth]{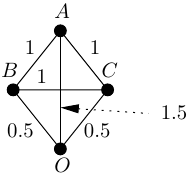}
\end{center}
\caption{Counter example of Lemma~\ref{lem:triangle} for non-Ptolemaic distances.}
\label{fig:abco}
\end{figure}

\begin{example}
\label{example:nonptolemaic}
Consider a distance matrix $d$ based on the index set $\{A,B,C,O\}$ such that for $X\neq Y$ (see also Figure~\ref{fig:abco})
$$
d(X,Y)=
\left\{
\begin{array}{ll}
 1 & \mbox{if }\{X,Y\}\subseteq \{A,B,C\}\\
 0.5 & \mbox{if }\{X,Y\}\in\{\{O,B\},\{O,C\}\}\\
 1.5 & \mbox{if }\{X,Y\}=\{O,A\}\\
\end{array}
\right.,
$$
One readily verifies that $d(O,A) = \frac{3}{2} (d(O,B) + d(O,C))$, which contradicts~\eqref{eq:ptolemaictriangle}.
\end{example}

In fact, multiplying the right-hand-side of~\eqref{eq:ptolemaictriangle} by 3/2, as in Example~\ref{example:nonptolemaic}, is enough for any distance matrix.
\begin{lemma}\label{lem:triangle_any_distance}
Let $(A,B,C)$ be a triplet such that $\dist{B}{C}\geq \max\{\dist{A}{B},\dist{A}{C}\}$ and $\{O\}$ index another coordinate of $d$. Then, 
$$ d(O,A) \leq \frac{3}{2} (d(O,B) + d(O,C)).$$
\end{lemma}

\begin{proof}
Using $\dist{B}{C}\geq \max\{\dist{A}{B},\dist{A}{C}\}$ and applying the triangle inequality at each step, we get 
\begin{align*}
d(O,A) \leq d(O,B) + \dist{A}{B} \leq d(O,B) + \dist{B}{C} \leq 2d(O,B)+d(O,C)\\
d(O,A) \leq d(O,C) + \dist{A}{C} \leq d(O,C) + \dist{B}{C} \leq d(O,B) + 2d(O,C) 
\end{align*}
Adding the above two inequalities provides the result.
\end{proof}

Using the above result, we can obtain a weaker counterpart of Theorem~\ref{thm:constant_ratio_general} for non-Ptolemaic distances.

\begin{theorem}\label{thm:constant_ratio_non_pto}
For any distance matrix $d$, $\costmax(G)\leq 9\:\cost_d(G)$.
\end{theorem}
\begin{proof}
 The proof follows exactly the approach followed in the proof of Theorem~\ref{thm:constant_ratio_general}, but we use Lemma~\ref{lem:triangle_any_distance} instead of Lemma~\ref{lem:triangle} when needed. 
We thus use the same notations as in the proof of Theorem~\ref{thm:constant_ratio_general}.
As $u_j^1u_j^2$ is a diameter of $\U_j$, we have $d(u_j^1,u_j^2) \ge \max(d(u_j^1,\bar{u}_j),d(u_j^2,\bar{u}_j))$, and we can apply Lemma~\ref{lem:triangle_any_distance} twice to triplet $(u_j^1,u_j^2,\bar{u}_j)$ to get
\begin{equation*}
\left\{\begin{aligned}
 \dist{u_i^1}{u_j^1} + \dist{u_i^1}{u_j^2} &\geq \frac{2}{3}\dist{u_i^1}{\bar{u}_j}\\
 \dist{u_i^2}{u_j^1} + \dist{u_i^2}{u_j^2} &\geq \frac{2}{3}\dist{u_i^2}{\bar{u}_j}.
\end{aligned}\right.
\end{equation*}
One last application of Lemma~\ref{lem:triangle_any_distance} in $u_i^1u_i^2\bar{u}_i$ then yields
$$\dist{u_i^1}{\bar{u}_j}+\dist{u_i^2}{\bar{u}_j}\geq \frac{2}{3}\dist{\bar{u}_i}{\bar{u}_j}=\frac{2}{3}\dmax_{ij}.$$
Summarizing the above, we get 
to $$\cost(G)\geq \costtilde(G) \ge \sum_{\{i,j\}\in E[G]} \Ex \left[\dist{\tilde{u}_i}{\tilde{u}_j}\right] \ge \frac{1}{4}\frac{4}{9}\costmax(G)= \frac{\costmax(G)}{9}.$$
\end{proof}

 \subsection{Bounding the approximation ratio on specific structures}
 \label{sec:specificgraphs}
 
 In what follows, we assume that the structure of the subragph induced by $G$ can be leveraged to obtain stronger bounds than in the previous section.
 We first describe how graph decomposition can be used to obtain such bounds.
 We then address the special graphs that have been singled out in our introductory applications, namely: paths, cycles, trees (subway network design), cliques (clustering), and stars (facility location). Unless stated otherwise, we assume throughout the section that $d$ is an arbitrary distance matrix, non-necessarily Ptolemaic.
 
\subsubsection{Building blocks}
\label{sec:BB}

We study below how the bounds obtained for distinct subgraphs can be combined to obtain a bound on their union.
\begin{proposition}\label{prop:graph-union}
Let $\GG$ be a graph, $G_t$ be a subgraph of $\GG$, and $\rho_t\geq 1$ such that $\costmax(G_t) \leq \rho_t \:\cost(G_t)$ for each $t=1,\ldots,T$. Then:
\begin{itemize}
\item $\costmax\left(\cup_{t=1}^TG_t\right) \leq T\times\max\limits_{t=1,\ldots,T}\rho_t\:\cost\left(\cup_{t=1}^TG_t\right)$, and
\item $\costmax(\cup_{t=1}^TG_t) \leq \max\limits_{t=1,\ldots,T}\rho_t\:\cost(\cup_{t=1}^TG_t)$ if, in addition,  $V(G_t)\cap V(G_{t'})=\emptyset$ for each $t\neq t'\in [T]$.
\end{itemize}

\end{proposition}
\begin{proof}
Let $G = \cup_{t=1}^TG_t$. 
In the first case we have
 $\displaystyle
 T\cdot\cost(G) \geq \sum_{t=1}^T \cost(G_t) \geq \sum_{t=1}^T \frac{1}{\rho_t} \costmax(G_t)
 \ge \frac{\costmax(G)}{\max\limits_{t=1,\ldots,T}\rho_t}.
 $

In the second case we have
  $\displaystyle
 \cost(G)
 = \sum_{t=1}^T \cost(G_t)
 \geq \sum_{t=1}^T \frac{1}{\rho_t} \costmax(G_t)
 \geq \frac{\costmax(G)}{\max\limits_{t=1,\ldots,T}\rho_t}.
 $
\end{proof}

The above results are particularly useful when combining sets $G_t$ having low values of $\rho_t$. The simplest example of such a set is a single edge.
\begin{observation}
\label{obs:edge}
For any $e\in E[G]$, $\costmax(e)=\cost(e)$.
\end{observation}
From the above observation and Proposition~\ref{prop:graph-union}, we obtain immediately that matchings also satisfy the equality.
\begin{corollary}
For any matching $G$, $\costmax(G)=\cost(G)$.
\end{corollary}
Matchings can be further combined to obtain general bounds that depend on the characteristics of $G$.
In the remainder, $\chi(G)$ denotes the edge chromatic number of $G$ and $\Delta(G)$ denotes its maximum degree.
\begin{corollary}\label{cor:block-coloring}
$\costmax(G)\leq\chi(G)\:\cost(G)$.
\end{corollary}
\begin{proof}\revision{
  By definition of $\chi(G)$, there exists a partition of $E[G]$ into $E_1,\dots,E_{\chi(G)}$ such that for any $t \in [\chi(G)]$,
  $E_t$ is a matching. For any $t \in [\chi(G)]$, we define $G_t$ such that $G_t$ only contains edges of $E_t$ (and no other vertices of edges).
  We have $G=\bigcup_{t=1}^{\chi(G)}G_t$, and by previous corollary on matching, $\costmax(G_t) \le c(G_t)$ for ant $t$.
  Thus, by Proposition~\ref{prop:graph-union}, we get the desired inequality.}
\end{proof}

\subsubsection{Graphs with small maximum degree}
Recall that Vizing's theorem states that $\chi(G)\leq \Delta(G)+1$. Combining this with Corollary~\ref{cor:block-coloring} implies that $\costmax(G)\leq (\Delta(G)+1)\cost(G)$. Actually, the bound can be decreased to $\Delta(G)$, as stated below.

\begin{proposition}\label{prop:delta}
For any graph $G$, $\costmax(G) \le \Delta(G)\:\cost(G)$.
\end{proposition}
\begin{proof}
 We follow the proof of Theorem~\ref{thm:constant_ratio_general} (and re-use same notations), but using different sets $\widetilde{\U}_i$.    
For any edge $\{i,j\}$ of $G$, define $u^j_i \in \U_i$ and $u^i_j \in \U_j$ such that $\dist{u^i_j}{u^j_i}=d^{max}_{ij}$.
For any vertex $i$, let $\widetilde{\U}_i=\{u_i^\ell, \ell \in N(i)\}$, and $\Delta_i=|N(i)|$ be the degree of $i$.
Let $\costtilde(G) = \max_{\uu\in\widetilde{\UU}}\sum_{\{i,j\}\in E[G]} \dist{u_i}{u_j}$. 
As $\widetilde{\UU} \subseteq \UU$, we know that $\cost(G) \ge \costtilde(G)$.
Let us now prove that $\costtilde(G) \ge \frac{\costmax(G)}{\Delta(G)}$.

For any vertex $i$, consider random variable $\tilde{u}_i$ taking any value of $\widetilde{\U}_i$ with equal probability. As in the proof of Theorem~\ref{thm:constant_ratio_general}, it is enough to lower bound the following quantity for an arbitrary edge $\{i,j\}\in E[G]$:
\begin{align*}
  \Ex \left[\dist{\tilde{u}_i}{\tilde{u}_j}\right]=&\frac{1}{\Delta_i\Delta_j}\left( \sum\limits_{\ell_1 \in N(i), \ell_2 \in N(j)}\dist{u_i^{\ell_1}}{u_j^{\ell_2}} \right).
\end{align*}
Suppose without loss of generality that $\Delta_j \leq \Delta_i$ and let $X_i=N(i)\setminus\{j\}$ and $X_j=N(j)\setminus\{i\}$ ($X_i$ or $X_j$ may be empty).
Let us define an arbitrary injective mapping $\phi:X_j\rightarrow X_i$, so that
\begin{equation}
 \label{eq:injective}
 \sum\limits_{\ell_1 \in X_i}\sum_{\ell_2 \in X_j}\dist{u_i^{\ell_1}}{u_j^{\ell_2}}\geq \sum_{\ell_1\in X_j}\sum_{\ell_2 \in X_j}\dist{u_i^{\phi(\ell_1)}}{u_j^{\ell_2}}\geq \sum_{\ell_2 \in X_j}\dist{u_i^{\phi(\ell_2)}}{u_j^{\ell_2}}.
\end{equation}
Using~\eqref{eq:injective}, we have
\begin{align*}
  \sum\limits_{\ell_1 \in N(i) \ell_2 \in N(j)}\dist{u_i^{\ell_1}}{u_j^{\ell_2}} &\ge d^{max}_{ij}+  \sum\limits_{\ell \in X_j}\dist{u_i^j}{u_j^{\ell}}+ \sum\limits_{\ell_1 \in X_i, \ell_2 \in X_j}\dist{u_i^{\ell_1}}{u_j^{\ell_2}}+  \sum\limits_{\ell \in X_i}\dist{u_i^{\ell}}{u_j^{i}} \\
  &\ge  d^{max}_{ij}+  \sum\limits_{\ell \in X_j}\dist{u_i^j}{u_j^{\ell}}+ \sum\limits_{\ell\in X_j}\dist{u_i^{\phi(\ell)}}{u_j^{\ell}}+  \sum\limits_{\ell \in X_i}\dist{u_i^{\ell}}{u_j^{i}} \\
  &\ge  d^{max}_{ij}+  \sum\limits_{\ell \in X_j}\dist{u_i^j}{u_j^{\ell}}+ \sum\limits_{\ell\in X_j}\dist{u_i^{\phi(\ell)}}{u_j^{\ell}}+  \sum\limits_{\ell \in X_j}\dist{u_i^{\phi(\ell)}}{u_j^{i}} \\
  &=  d^{max}_{ij}+  \sum\limits_{\ell \in X_j}\left(\dist{u_i^j}{u_j^{\ell}}+ \dist{u_i^{\phi(\ell)}}{u_j^{\ell}}+ \dist{u_i^{\phi(\ell)}}{u_j^{i}}\right) \\
  &\ge  d^{max}_{ij}+  \sum\limits_{\ell \in X_j}d^{max}_{ij} = \Delta_{j}d^{max}_{ij}.
\end{align*}
We obtain
$$
  \Ex \left[\dist{\tilde{u}_i}{\tilde{u}_j}\right]\ge\frac{1}{\Delta_i\Delta_j}\Delta_{j}d^{max}_{ij}=\frac{1}{\Delta_i}d^{max}_{ij} \ge \frac{1}{\Delta(G)}d^{max}_{ij}.
$$
\end{proof}

Proposition~\ref{prop:delta} immediately implies the following Corollary, which had an ad-hoc proof for cycles in~\cite{citovsky2017tsp}.
\begin{corollary}\label{cor:ratio_paths_cycles}
  Let $G$ be a path or a cycle, then $\costmax(G)\leq 2\:\cost(G)$.
\end{corollary}

We verify in the two propositions below that the above bound is tight.
\begin{proposition}\label{prop:tight-paths}
For any path $G$ of length at least three, there exists an uncertainty set $\U$ such that $\costmax(G) = 2\:\cost(G)$.
\end{proposition}
\begin{proof}
 We consider a path $G$ composed of $n\geq 3$ vertices, where $E[G] = \{\{i,i+1\} : i\in\{1,\dots,n-1\}.$
The vertices are located on a one-dimensional line where $\U_1=\{0\}$, $\U_2=\{0,1\}$ and $\U_3=\dots=\U_n=\{1\}$.
We have $\dmax_{12}=\dmax_{23}=1$ and $\dmax_{i,i+1}=0,\forall i=3,\dots,n-1$, so $\costmax(G) = 2$. 
There are only two feasible solutions depending on whether $u_2=1$ or $u_2=0$, and they have the same cost $\cost(G)=1$.
\end{proof}
\begin{proposition}\label{prop:tight-cycles}
For any cycle $G$ of length at least four, there exists an uncertainty set $\U$ such that $\costmax(G) = 2\:\cost(G)$.
\end{proposition}
\begin{proof}
 We consider a cycle $G$ composed of $n\geq 4$ vertices, where $E[G]=\{\{i,i+1\} : i\in\{1,\dots,n-1\}\cup \{\{n,1\}\}.$
The vertices are located on a one-dimensional line where $\U_1=\{0\}$, $\U_2=\{0,1\}$, $\U_3=\{1\}$, $\U_4=\{0,1\}$ and $\U_5=\dots=\U_n=\{0\}$ if $n\geq 5$.
We verify that $\costmax(G) = 4$, and there are four feasible solutions depending on whether $u_2=1$ or $u_2=0$ and $u_4=1$ or $u_4=0$. 
These four solutions all have the same cost $\cost(G)=2$.
\end{proof}

The above proposition considers cycles that contain at least four vertices, so one can wonder what happens in the case of smaller cycles. We show next that for cycles that contain only 3 vertices, the bound can be reduced to $3/2$.

\begin{remark}
Consider the 3-cycle $G$.
Applying the triangle inequality three times yields:
\begin{align*}
    \cost(G) = \max_{u\in\UU} \left(d(u_1,u_2) + d(u_2,u_3) + d(u_1,u_3)\right) & \geq \max_{u\in\UU} \left(2\max\{d(u_1,u_2),d(u_2,u_3),d(u_1,u_3)\}\right)\\
        &\geq 2\max\{\dmax_{12},\dmax_{23},\dmax_{13}\}\\
        &\geq \frac{2}{3}\costmax(G),
\end{align*}
so the maximum worst-case factor is bounded by 3/2. 
This bound is tight.
To see this, one can look at the case exhibited in the proof of Proposition~\ref{prop:tight-paths} ($\U_1=\{0\}$, $\U_2=\{0,1\}$ and $\U_3=\{1\}$). 
For the 3-cycle, $\costmax(G) = 3$ and $\cost(G)=2$.
\end{remark}

\subsubsection{Cliques}

 We now turn to the special case where $\G(I)$ contains only cliques. 
 One specificity of a clique $G$ is that for any matching $M$ of size $\floor{\frac{n}{2}}$, every edge of $G$ belongs to a triangle including one edge of $M$. 
 Applying the triangle inequality repeatedly for a well chosen matching provides the following ratio.
 \begin{proposition}
 \label{prop:cliques}
  Let $G$ be a clique. Then, $\costmax(G) \leq 2\:\cost(G)$.
 \end{proposition} 
 \begin{proof}
   Recall that $n$ denotes $|V[G]|$. It is folklore that $\chi(G) = n-1$ when $n$ is even, and $\chi(G)=n$ when $n$ is odd, leading to $\chi(G) \le n$ for any $n$. This implies that $E[G]$ can be partitioned into $n$ matchings $M_i$, each of size  $\lfloor n/2\rfloor$, and thus that $\costmax(G)=\sum_{i=1}^n \costmax(M_i)$. Therefore, there is a matching of $G$, denoted $M^*$, such that
  \begin{equation}
  \label{eq:mprime}
   \costmax(M^*)\geq \frac{1}{n} \costmax(G).
  \end{equation}
  Next, we define $\uu^*$ as any element from $\argmax\limits_{\uu\in\UU}\sum\limits_{\{i,j\}\in M^*}d(u_i,u_j)$, and we use the shorter notation $d^*_{ij}=d(u^*_i,u^*_j)$, and $d^*_e = d^*_{ij}$ for
  any edge $e = \{i,j\}$. Observe that because $M^*$ is a matching
  \begin{equation}
  \label{eq:dequalityMstar}
   d^*_{e} = d^{max}_{e}
  \end{equation}
for each $e\in M^*$. For any $E' \subseteq E[G]$, let $d^*(E')=\sum_{e \in E'}d^*_e$. Observe that $\cost(G) \ge d^*(G)$.
Our objective is to prove that $d^*(G) \ge \frac{1}{2}\costmax(G)$.

  Assume without loss of generality that $M^* = \{\{2i-1,2i\} \mid i \in [\lfloor \frac{n}{2} \rfloor ]\}$, so $n$ is the only vertex not belonging to any edge of $M^*$ when $n$ is odd. For any  $i \in [\lfloor \frac{n}{2} \rfloor ]$, let $X(2i-1,2i) = \{\{2i-1,l\} \cup \{2i,l\} \mid l \in V[G] \setminus \{2i-1,2i\} \}$.
  As a consequence, $X(2i-1,2i)\cup\{2i-1,2i\}$ gathers all the edges of the triangles of $G$ including $\{2i-1,2i\}$. Observe that the triangle inequality yields
  \begin{equation}
  \label{eq:mprime2}
  d^*(X(2i-1,2i)) = \sum_{l \in V[G] \setminus \{2i-1,2i\}}(d^*_{(2i-1)l}+d^*_{(2i)l}) \ge (n-2)d^*_{(2i-1)(2i)}. 
  \end{equation}
  Summing up~\eqref{eq:mprime2} for all $i \in [\lfloor \frac{n}{2} \rfloor ]$, we obtain
  \begin{equation*}
  \label{eq:mprime3}
  \sum_{i \in [\lfloor \frac{n}{2} \rfloor ]}d^*(X(2i-1,2i)) \ge (n-2)d^*(M^*)=(n-2)\costmax(M^*),
  \end{equation*}
  where the last equality follows from~\eqref{eq:dequalityMstar}. What is more, any edge $e \in E[G] \setminus M^*$ belongs to at most two sets $X(2i-1,2i)$, so that
  \begin{equation*}
  \label{eq:mprime4}
  2 d^*(E[G] \setminus M^*) \ge \sum_{i \in [\lfloor \frac{n}{2} \rfloor ]}d^*(X(2i-1,2i)).
  \end{equation*}
  We obtain
  \begin{equation*}
  \label{eq:mprime5}
  2d^*(G) = 2 \bigg[d^*((E[G] \setminus M^*))+d^*(M^*)\bigg] \ge (n-2)\costmax(M^*) + 2\costmax(M^*) = n\costmax(M^*).
  \end{equation*}
  From~\eqref{eq:mprime}, we finally get $\cost(G)\geq d^*(G)\geq \frac{n}{2}\costmax(M^*)\geq \frac{1}{2}\costmax(G)$.
\end{proof}
 
We show below that the above bound is asymptotically tight, even for very simple distances.

  \begin{proposition}\label{prop:tight-cliques}
 If $G$ is a $k$-clique, there exists an uncertainty set $\U$ such that $\costmax(G) = \frac{2(k-1)}{k}\:\cost(G)$ if $k$ is odd and $\costmax(G) = \frac{2k}{k+1}\:\cost(G)$ if $k$ is even.
 \end{proposition}
 \begin{proof}
 Let us define $\U_i = \{0,1\}$ for any $i \in V[G]$, and the distance $d$ by $d(x,y)=|x-y|$. Hence, $\costmax(G)=m$, where $m=|E[G]|=\frac{n(n-1)}{2}$.
 
 On the other hand, we see that computing $\cost(G)$ is equivalent to partitioning $V$ into $\{V_1,V_2\}$, and defining $u_i = 0$ if $i \in V_1$, and $1$ if $i \in V_2$. Hence, $\cost(G)$ is equal to the optimal cost of a max-cut for $G$. That cost is equal to $\frac{n^2}{4}$ and $\frac{(n-1)}{2}\frac{(n+1)}{2}$ when $n$ is even and odd, respectively, leading in both cases to the claimed ratio.
\end{proof}

\subsubsection{Stars}
 
In what follows, we consider stars whose center is vertex $1$ (meaning that for any $i \neq 1$, $|N(i)|=1$).
 
\begin{proposition}
\label{prop:stars}
  Let $G$ be a star. Then, assuming that $d$ is Ptolemaic, $\costmax(G) \leq 2\:\cost_d(G)$.
\end{proposition}
\begin{proof}
 Let $\{u_1^1, u_1^2\}$ such that $\dist{u_1^1}{u_1^2}=\diam(\U_1)$. Let $\bar{u}_1\in\U_1$ and For all $i\in\{2,\dots,n\}$, let $\bar{u}_i\in\U_i$ such that there exists $\bar{u}_1\in\U_1$ with $d(\bar{u}_1,\bar{u}_i)= \dmax_{1i}$. 
 We then follow the same approach as in the proof of Theorem~\ref{thm:constant_ratio_general}: we set $\overline{\U}_1 = \{u_1^1,u_1^2\}$,  $\overline{\U}_i=\{\bar{u}_i\}, \forall i >1,$ and for $i\in V$, we consider the random variables $\tilde{u}_i$ uniformly distributed on $\U_i$. 
 Since $\{u_1^1, u_1^2\}$ is a diameter of $\U_1$, we can apply Lemma~\ref{lem:triangle} to get  $$d(u_1^1,\bar{u}_i) + d(u_1^2,\bar{u}_i) \geq d(\bar{u}_1,\bar{u}_i)= \dmax_{1i}.$$ This implies  $\Ex \left[\dist{\tilde{u}_1}{\tilde{u}_i}\right] \ge \frac{1}{2}d^{max}_{1i}$ for any $i > 1$, and thus the claimed ratio.
\end{proof}

\begin{corollary}\label{cor:star}
  Let $G$ be a star. Then, $\costmax(G) \leq 3\:\cost(G)$.
\end{corollary}
\begin{proof}
 The result is obtained with the exact same proof as Proposition~\ref{prop:stars} where we apply Lemma~\ref{lem:triangle_any_distance} instead of Lemma~\ref{lem:triangle}.
\end{proof}

We show below that the bound from Corollary~\ref{cor:star} is asymptotically tight.

\begin{proposition}
\label{prop:tight-stars}
  Let $G$ be a star on $n$ vertices.
There is an uncertainty set $\UU$ such that $\costmax(G) = \frac{3(n-1)}{n+1}\:\cost(G)$.
\end{proposition}
\begin{proof}
 We consider the uncertainty set $\UU = \varprod_{i=1}^{n}\U_i$ where $\U_1=\{u_1^2,\dots,u_1^{n}\}$ and $\U_i=\{u_i\},i=2,\dots,n,$ such that: 
\begin{itemize}
    \item for all $(i,j)\in \{2, \dots,n\}^2, i\neq j$, $d(u_1^i,u_1^j)= 2/3$,
    \item for all $(i,j)\in \{2, \dots,n\}^2, i\neq j$, $d(u_i,u_j) = 2/3$,
    \item for all $i=2,\dots,n$, $d(u_1^i,u_i)= 1$ and $\forall j\neq i$, $d(u_1^j,u_i) = 1/3$.
\end{itemize}
The triangle inequality is verified, so $d$ is a distance.

By symmetry of the star graph and of the uncertainty set, every solution $u\in\U$ is optimal and has the same value $1+\frac{1}{3}(n-2)$. 
Moreover, the maximum distance between $\U_1$ and $\U_i$ is equal to $1$ for all $i\in\{2,\dots,n\},$ so $\costmax(G)=n-1$.
As a result, $\costmax(G) = \frac{3(n-1)}{n+1}\cost(G)$.
\end{proof}

\subsubsection{Trees}

We conclude our study of specific structures with trees.
Our first result combines the bounds obtained for stars in the previous section with the composition results presented in Section~\ref{sec:BB} to improve the ratio of $9$ obtained in Theorem~\ref{thm:constant_ratio_non_pto} for general graphs and distances.

\begin{proposition}\label{prop:tree}
  Let $G$ be a tree. Then $\costmax(G) \leq 6\:\cost(G)$.
  \end{proposition}
\begin{proof}
   Observe first that we can partition $E[G]$ into $E_1$ and $E_2$ such that each $E_i$ induces a star forest (a graph where any connected component is a star).
  Indeed, to obtain such a partition we root the tree at vertex $1$, and define $S^i$ as the star whose central vertex is $i$ and whose leaves are the children vertices of $i$ in $G$. 
  Then, we define $E_1$ (resp. $E_2$) as the union of $E(S^i)$ for vertices $i$ which are an odd (resp. even) distance from vertex $1$.
  By Proposition~\ref{prop:graph-union} and Corollary~\ref{cor:star}, we get the claimed ratio.
\end{proof}

Next, we show how the bound can be further tightened when considering distances resulting from Euclidean embeddings and spherical uncertainty sets. 
 
\begin{proposition}\label{prop:tree_eucl}
Assume $d$ is the Euclidean distance, and that for all $i\in V$, $\U_i$ is a closed ball.
Then, for any tree $G$, $\costmax(G) \leq 2\sqrt{2}\:\cost(G)$. 
\end{proposition}
\begin{proof}
 We assume without loss of generality that $G$ is rooted at vertex $1$. 
The proof is made by induction on the height of the tree. 
For this, we consider the following induction statement:

\begin{quote}
${{P}(h)}:$
\emph{
If $G$ has height $h$, there are two solutions $u^1,u^2\in \UU$ such that
\begin{itemize}
    \item $\costmax(G) \leq 2\sqrt{2}\:\sum_{\{i,j\}\in E[G]} d(u_i^1,u_j^1)= \sum_{\{i,j\}\in E[G]} d(u_i^2,u_j^2)$, 
    \item $u_i^1=u_i^2$ for each vertex $i$ with level $l<h$,
    \item $[u_i^1,u_i^2]$ is a diameter of $\U_i$ for each vertex $i$ with level $h$.
\end{itemize} 
} 
\end{quote}

If $h=0$, $G$ has only one vertex which is the root of the tree, so the induction statement is trivially satisfied with any diameter $[u^1,u^2]$ of $\U_1$. 

Assume that $P(h)$ is true for some $h\geq 1$, and let $G$ be a rooted tree with height $h+1$. 
Without loss of generality, we assume that the vertices of $V[G]$ are sorted by increasing level and let $n_{h-1}$, $n_h$ and $n_{h+1}$ be such that the vertices with level $h$ are $V_h = \{n_{h-1}+1, \dots, n_h\}$ and those with level $h+1$ are $V_{h+1}=\{n_{h}+1, \dots, n_{h+1}\}$.
Let $G_{\leq h}=(V_{\leq h},E_{\leq h})$ be the subtree of $G$ induced by vertices $\{1,\dots,n_h\}$.
Tree $G_{\leq h}$ has height $h$ so we can apply the induction statement to get two solutions $u^1,u^2\in\UU$ as described in ${P(h)}$. 

Now, for all $i\in V_h$, let $S_i\subset E[G]$ be the set of edges of the star graph whose internal vertex is $i$ and whose leaves are the children vertices of $i$ (which all belong to $V_{h+1}$).
Similarly to what was done in the proof of Proposition~\ref{prop:stars}, for all $\{i,j\}\in S_i$, we can set $u_j^1\in\argmax_{u\in\U_j} d(u,u_i^1)$ and $u_j^2\in\argmax_{u\in\U_j} d(u,u_i^2)$, which yields $d(u_i^1, u_j^1)+d(u_i^2,u_j^2)\geq \dmax_{ij}$.
We can then assume without loss of generality that $\sum_{\{i,j\}\in S_i} d(u_i^1, u_j^1)\geq \sum_{\{i,j\}\in S_i} d(u_i^2, u_j^2)$, i.e., \begin{equation}\label{eq:intermed_balltree1}
    \costmax(S_i)\leq 2\sum_{\{i,j\}\in S_i} d(u_i^1, u_j^1).
\end{equation}

Given that we are considering the Euclidean distance with spherical uncertainty sets, the above implies that segment $[u_i^1,u_j^1]$ goes through the center $o_j$ of $\U_j$ (direct application of the triangle inequality).
Then, let $[\bar{u}_j^1, \bar{u}_j^2]$ be the diameter of $\U_j$ that is orthogonal to $[u_i^1,u_j^1]$; it exists because $\U_j$ is spherical.
Denoting the radius of $\U_j$ by $r_j$, we then compute
\begin{equation*}
    \left\{\begin{aligned}
    &d(u_i^1,u_j^1)^2 = d(u_i^1,o_j)^2 + r_j^2 + 2 r_j d(u_i^1,o_j) \leq 2 (d(u_i^1,o_j)^2 + r_j^2)\\
    &d(u_i^1,\bar{u}_j^1)^2 =d(u_i^1,\bar{u}_j^2)^2 =  d(u_i^1,o_j)^2 + r_j^2
    \end{aligned}\right.
\end{equation*}
As a consequence, $d(u_i^1,\bar{u}_j^1)=d(u_i^1,\bar{u}_j^2)\geq 1/\sqrt{2} d(u_i^1,u_j^1)$.
Using~\eqref{eq:intermed_balltree1}, we then get
$$\costmax(S_i)\leq 2\sqrt{2}\sum_{\{i,j\}\in S_i}  d(u_i^1,\bar{u}_j^1) = 2\sqrt{2}\sum_{\{i,j\}\in S_i}  d(u_i^1,\bar{u}_j^2).$$
The same applies to all $i$, so we build two solutions $\tilde{u}^1,\tilde{u}^2\in\UU$ such that 
\begin{itemize}
\item $\tilde{u}_i^1=\tilde{u}_i^2=u_i^1$ for all $i\in V_{\leq h}$,
\item $[\tilde{u}_i^1,\tilde{u}_i^2]=[\bar{u}_i^1,\bar{u}_i^2]$ is a diameter of $\U_i$ for all $i\in V_{h+1}.$ 
\end{itemize}
To conclude, we observe that $E[G]=E_{\leq h}\cup \left(\bigcup_{i\in V_h} S_i \right)$, so 
\begin{align*}\sum_{\{i,j\}\in E[G]} d(\tilde{u}_i^1,\tilde{u}_j^1) &= \sum_{\{i,j\}\in E_{\leq h}} d(\tilde{u}_i^1,\tilde{u}_j^1) + \sum_{i\in V_h}\sum_{\{i,j\}\in S_i} d(\tilde{u}_i^1,\tilde{u}_j^1)\\
    & = \sum_{\{i,j\}\in E_{\leq h}} d(u_i^1,u_j^1) + \sum_{i\in V_h}\sum_{\{i,j\}\in S_i} d(u_i^1,\bar{u}_j^1)\\
    &\geq \frac{1}{2\sqrt{2}} \costmax(G_{\leq h}) +  \frac{1}{2\sqrt{2}}\sum_{i\in V_h}\costmax(S_i)\\
    & = \frac{1}{2\sqrt{2}} \costmax(G).
\end{align*}
\end{proof}

\begin{remark}
 Notice that Proposition~\ref{prop:tree_eucl} does not fit into the framework described so far as the sets $\U_i$ are infinite. Nevertheless, a finite input for such problems can be defined as $I^{ball}=(\GG,\additionalInput,c,r,\ell)$, where $\GG$ and $\additionalInput$ are as before, $c$ and $r$ denote the center and radii of the balls, respectively, and $\ell$ denotes their dimension. Furthermore in such problems, we consider that the distance between any two points in $u,v\in\R^\ell$ can be computed in constant time by calling the function $\|u-v\|_2$.
\end{remark}

\section{FPTAS for robust shortest path}
\label{sec:FPTAS}

Up to now, we have mostly focused on general problems \ROBUST{$\Pi$} and provided constant-factor approximation algorithms. The purpose of this section is to focus on a specific problem, \ROBUST{\DELTASP}, known to be $\NP$-hard~(\cite{BougeretOmerPoss1}).
We do this by providing a dynamic programming algorithm for \ROBUST{\SP} (a generalization of \ROBUST{\DELTASP} where $d$ is not required to verify the triangle inequality) in Section~\ref{section:DP}.
\revision{As observed at the end of Section~\ref{section:DP}, this dynamic porgramming algorithm becomes polynomial when some parameters are constant (typically $|\U|$). Moreover, it is also used in Section~\ref{sec:fptas} to derive the \FPTAS.}

\subsection{Dynamic programming algorithm}
\label{section:DP}

In the \ROBUST{\SP} problem, the input $I = (\GG,s,t,\U,\drob)$ is the same as in the \ROBUST{\DELTASP} problem, except that $\drob$ is now a non-negative matrix that is only assumed to be symmetric, e.g., $\drob(u,v)=\drob(v,u)$ for any $u,v\in\bigcup_{i\in \VV}\U_i$.
In particular, the function associated with $\drob$ may not respect the triangle inequality. This level of generality will be useful in the next section when deriving the \FPTAS for \ROBUST{\DELTASP}.



Next, we introduce further notations that are used to derive the dynamic programming algorithm. Let $\itpaths{i}$ denote the set of all $i-t$ simple paths in $\GG$, and $\itpathsless{i}{\kappa}$ those having at most $\kappa$ edges. 
Given $i \in \VV$ and a path $P\in\itpaths{i}$, we define the worst-case cost given that $u_i=u_i^\ell$ for $\ell\in[\nU{i}]$ (recall that $\nU{i}=\abs{\UU_i}$) as
\begin{itemize}
  \item $c^\ell(P)=\max \{c(u,P) \mid u \in \UU, u_i = u_i^\ell \}$, and
  \item and $Pr(P) = (c^\ell(P))_{\ell \in [\nU{i}]}$ as the profile of $P$.
\end{itemize}
We also introduce for any $\kappa \in [n]$ the set of profiles of all $i-t$ paths with at most $\kappa$ edges as $\P^{(i,\kappa)} = \{Pr(P) \mid P\in \itpathsless{i}{\kappa}\}$.
We denote
\begin{itemize}
\item $\Val(I) = \{c^\ell(P) \mid i \in \VV, P\in\itpaths{i}, \ell \in [\nU{i}]\},$
\item $\nval = |\Val(I)|$, and
\item $\nprof = |\bigcup_{i \in \VV, \kappa \in [n]}\P^{(i,\kappa)}|$ the total number of different profiles. 
\end{itemize}

Our objective is to define an algorithm $A(i,\kappa)$ that, given any $i \in \VV$ and $\kappa \in [n]$ computes $(\P^{(i,\kappa)},Q^{(i,\kappa)})$ such that
\begin{itemize}
\item $Q^{(i,\kappa)}\subseteq \itpathsless{i}{\kappa}$,
\item $|Q^{(i,\kappa)}|=|\P^{(i,\kappa)}|$,
\item for any $p \in \P^{(i,\kappa)}$, there exists $P \in Q^{(i,\kappa)}$ such that $Pr(P)=p$.
\end{itemize}
Informally, $A(i,\kappa)$ computes all profiles associated to $(i,\kappa)$ as well as a representative path for each one of these profiles. Let us first verify that computing this is enough to solve \ROBUST{\SP} optimally.
\begin{lemma}\label{lemma:laststep}
\revision{Given an input $I = (\GG,s,t,\U,\drob)$ of the \ROBUST{\SP} problem, and given $(\P^{(s,n)},Q^{(s,n)})$, we can find an optimal solution of $I$ in time polynomial in $n$ and linear in $\nprof$.}
\end{lemma}
\begin{proof}
For any $p \in \P^{(s,n)}$, $p=(p_\ell)_{\ell \in [\nU{s}]}$, let $x_p = \max_{\ell \in [\nU{s}]} p_\ell$.
We define $$p^{min} =\argmin_{p \in \P^{(s,n)}} x_p$$ and output $P^{min} \in Q^{(s,n)}$ such that $Pr(P^{min})=p^{min}$.
Let $P^*$ be an optimal solution and $p^* = Pr(P^*)$. 
As  $p^* \in \P^{(s,n)}$, we have $c(P^{min})=x_{p^{min}} \le x_{p^*}=c(P^*)$.
\end{proof}

We provide next the dynamic programming recursion for $\P^{(i,\kappa)}$, leaving aside the computation of $Q^{(i,\kappa)}$ to simplify the presentation.
Given $i \in \VV$, $\kappa \in [n]$, $j \in N(i)$, $P'\in\itpathsless{j}{\kappa}$, and $p' = Pr(P')$, we consider the $i-t$ path $P = iP'$ obtained by concatenating $i$ with $P'$. One readily verifies that $Pr(P)=p(i,\kappa,j,p')$, where $p(i,\kappa,j,p')=(y_\ell)_{\ell \in [\nU{i}]}$, with $y_\ell = \max_{\ell' \in [\nU{j}]} \dist{u_i^\ell}{u_j^{\ell'}}+p'_{\ell'}$. We obtain that for any $i\neq s$ and $\kappa>0$
\begin{equation}
\label{eq:recursion}
\P^{(i,\kappa)} = \{p(i,\kappa,j,p') \mid j \in N(i), p' \in \P^{(j,\kappa-1)}\},
\end{equation}
and $\P^{(s,0)}=(0)_{\ell\in[\nU{s}]}$. Recall that $\NU=\max_{i \in \VV}\nU{i}$. We provide in the next lemma the complexity of the resulting dynamic programming algorithm.

\begin{lemma}\label{lemma:recursioncomplexity}
Let $i \in \VV$, $\kappa \in [n]$. Given $\P^{(j,\kappa-1)}$ and $Q^{(j,\kappa-1)}$ for any $j \in N(i)$, we can compute
$(\P^{(i,\kappa)},Q^{(i,\kappa)})$ in time $\grandO(n \NU^2 \nprof)$ and space $\grandO(\nprof n(log(n)))$.
\end{lemma}
\begin{proof}
 For complexity issues we assume that  $\P^{(j,\kappa-1)}$ and $Q^{(j,\kappa-1)}$ are represented as arrays $A_P$ and $A_Q$ indexed by profiles,
where given a profile $p$, $A_P[p]$ is true iff $p \in \P^{(j,\kappa-1)}$, and $A_Q[p]$ contains a path $P$ such that $Pr(P)=p$ if $p \in \P^{(j,\kappa-1)}$, and $\emptyset$ otherwise.
This explains the $\grandO(\nprof n(log(n)))$ required space for storing $Q$.
As profiles are vectors of length at most $\NU$, we consider that it takes $\grandO(\NU)$ to obtain the value stored at index $p$ of array $A_P$ or $A_Q$.

We compute $\P^{(i,\kappa)}$ following the recursion relation~\eqref{eq:recursion}, and compute $Q^{(i,\kappa)}$ along the way. More precisely, we start by initializing two arrays $A'_P$ and $A'_Q$ of size $\nprof$.
Then, for all $j \in N(i)$ and $p' \in \P^{(j,\kappa-1)}$, we compute $p(i,\kappa,j,p')$ in time $\grandO(\NU^2)$.
Now, we perform the following operations in $\grandO(1)$. 
If $p(i,\kappa,j,p')$ is not already in $\P$, we add it to $\P$, we find a path $P'$ in $Q^{(j,\kappa-1)}$ such that $Pr(P')=p'$, and we add the path $iP'$ to $Q$. 
\end{proof}

We are now ready to state the main result of this section.
\begin{theorem}\label{thm:dp}
\ROBUST{\SP} can be solved in time $\grandO(n^3  \NU^2 \nprof)$ and space $\grandO(\nprof n^3log(n))$.
\end{theorem}
\begin{proof}
\revision{We compute  $(\P^{(s,n)},Q^{(s,n)})$ using a DP algorithm based on~\eqref{eq:recursion}, and obtain an optimal solution by Lemma~\ref{lemma:laststep}.
The overall complexity is dominated by the DP, whose complexity is in $\grandO(a \times b)$, where $a$ denotes the number of entries of the associated memoization table,
and $b$ is the time-complexity of computing one entry, assuming the other ones are accessible in $\grandO(1)$.
As there is an entry $(i,\kappa)$ for any $i \in \VV$ and $\kappa \in [n]$, we get $a = \grandO(n^2)$.
By Lemma~\ref{lemma:recursioncomplexity}, we get $b=\grandO(n \NU^2 \nprof)$  get the claimed time complexity. The space complexity immediately follows.}
\end{proof}

Let us further elaborate on the value of $\nprof$ that arises in Theorem~\ref{thm:dp}. First of all, we see that $\nprof \le (\nval)^{\NU}$, leading to the observation below, used in the next section to derive the FPTAS. 
\begin{observation}
\label{obs:dp}
\ROBUST{\SP} can be solved in time $\grandO(n^3 \NU^2 (\nval)^{\NU})$
\end{observation}
From a more theoretical viewpoint, notice  that the reduction of \cite{BougeretOmerPoss1}[Proposition~1] proving the hardness of~\ROBUST{\SP} involves a ``large'' set $\U$, so a natural question is whether \ROBUST{\SP} becomes polynomial for ``small'' sets $\U$, either in terms of diameter or number of elements. It so happens that the two questions can be answered positively. Namely, observe that for any value $c^\ell(P)$, we can find some $n_{v,v'} \in [n]$ for each $(v,v')\in\U^2$ such that $c^\ell(P)=\sum_{(v,v') \in \U^2}n_{v,v'}\dist{v}{v'}$. Therefore, $\nval$ can be bounded by $n^{\frac{|\U|(|\U|-1)}{2}}$, and \ROBUST{\SP} can be solved in polynomial time if $|\U|$ is constant. Alternatively, if all distances are integer, meaning $d$ has integer values, then $\nprof\leq n\times \mbox{diam}(\U)$, so \ROBUST{\SP} can be solved in $\grandO(n^3 \NU^2 (n\times\mbox{diam}(\U))^{\NU})$ in that case. 

\subsection{FPTAS}\label{sec:fptas}

We now consider solving \ROBUST{\DELTASP} approximately. More precisely, we provide an algorithm that, given an input $I$ of \ROBUST{\DELTASP} and $\epsilon > 0$, outputs an $(1 + \epsilon)$-approximated solution in time polynomial both in $n$ and $1/\epsilon$.
More precisely, given any $\epsilon > 0$, we want to provide an $(1+\epsilon)$-approximated solution. Let $I=(\GG,s,t,\U,\drob)$ be an input to \ROBUST{\DELTASP} and $A$ be an upper bound to $\OPT(I)$ and $\epsilon' > 0$. We define a matrix $\drob'$ by rounding each element of $\drob$ to the closest value of the form $\epsilon'\ell A$ for some $\ell \in \mathbb{N}$. We obtain an instance $I'=(\GG,s,t,\U,\drob')$ to \ROBUST{\SP}. Having Observation~\ref{obs:dp} in mind, a straightforward application of the DP from the previous section to $I'$ would yield too many values $n_{val}$. Hence, we show in Section~\ref{sec:fptas} how to adapt the DP, and choose $A$ and $\epsilon'$ (depending on $\epsilon$) appropriately to obtain the result below.
\begin{theorem}\label{thm:fptas}
 \ROBUST{\DELTASP} admits an \FPTAS for fixed $\NU$: for any $\epsilon$, we can compute a $(1+\epsilon)$-approximated solution
 in time $\grandO(n^3 \NU^2 (\frac{n^2}{\epsilon})^{\NU})$.
\end{theorem}
\begin{proof}
 Let $I=(\GG,s,t,\UU,\drob)$ be an instance of \ROBUST{\DELTASP}.
 Our objective is to provide a solution of cost at most $(1+\epsilon)\OPT(I)$.
Let $\epsilon' \in \mathbb{R}^+$.
Using the $2$-approximation obtained combining Theorem~\ref{thm:approx} and Proposition~\ref{prop:delta}, we start by computing an $s-t$ path $P_A$ of cost $c(P_A)=A$, where $\OPT(I) \le A \le 2\OPT(I)$.       
For any $x,y \in  \bigcup_{i \in \VV}\UU_i$, we define $\drob'(x,y)$ by rounding up $\drob(x,y)$ to the closest value of the form $\epsilon'\ell A$ for some $\ell \in \mathbb{N}$.
For any path $P$ and $\uu \in \UU$ we denote by  $c'(\uu,P)=\sum_{{i,j} \in P}\drob'(u_i,u_j)$,  $c'(P)=\max_{\uu \in \UU}c'(\uu,P)$.
Let $I'=(\GG,s,t,\UU,\drob')$ be the instance of \ROBUST{\SP} obtained when using $\drob'$ instead of $\drob$.

Observe that
\begin{itemize}
\item the function $d'(u,v)=\drob'(u,v)$ may not be a distance, 
\item for any $x,y \in \bigcup_{i \in \VV} \UU_i$, we have $\drob(u,v) \le \drob'(u,v) \le \drob(u,v)+\epsilon' A$
\item for any path $P$, $c(P) \le c'(P) \le c(P)+n \epsilon' A$
\item $\OPT' \le \OPT+n \epsilon' A$.
\end{itemize}

Let $i \in \VV$ and $P$ be an $i-t$ path. We say that $P$ is \emph{useless} if it verifies $c'(P) > A(1+n \epsilon')$; otherwise, $P$ is said to be \emph{good}. According to previous observations, we see that $c'(P_A) \le A(1+n \epsilon')$.
Thus, for any $i \in \VV$ and useless $i-t$ path $P$, we have $c'(P) > c'(P_A)$. This implies that $P$ cannot be the suffix of an optimal solution to input $I'$ (meaning that there is no optimal solution of $I'$ that first goes from $s$ to $i$, and then uses $P$).
As a consequence, in the DP algorithm provided in Section~\ref{section:DP}, we can restrict our attention to the profiles of good paths, without loosing optimality in $I'$.
More formally, for any $i \in \VV$ and $\kappa \in [n]$, we adapt the previous definition $\P^{(i,\kappa)}$ to
$$
\P_{\mathrm{good}}^{(i,\kappa)}= \{Pr(P) \mid P\in \itpathsless{i}{\kappa}, P\mbox{ is good}\},
$$
and we now consider that the DP algorithm
$A_g(i,\kappa)$ computes $(\P_{\mathrm{good}}^{(i,\kappa)},Q^{(i,\kappa)})$ instead of $(\P^{(i,\kappa)},Q^{(i,\kappa)})$.
We now compute $P^*$ an optimal solution on instance $I'$ (for cost function $c'$) using Theorem~\ref{thm:dp} (with DP algorithm $A_g$), and output $P^*$.

We have $c(P^*) \le c'(P^*) = \OPT(I') \le \OPT(I)+n \epsilon' A \le \OPT(I)(1+2n\epsilon')$.
Let us now consider the complexity of computing $P^*$.
As $c'(P)=\max_{\ell \in [\nU{i}]}c^{\prime \ell}(P)$, observe that for any good path $P$ we have $c^{\prime \ell}(P) \le A(1+n\epsilon')$.
Moreover, as for any $x,y \in \bigcup_{i \in \VV}\UU_i$, $\drob'(u,v)$ is a multiple of $\epsilon' A$, we get that for any good path $P$,
$c^{\prime \ell}(P) = \ell \epsilon' A$ for $0 \le l \le n+\lceil \frac{1}{\epsilon'} \rceil$.
This implies that $\nval \le (n+2+\frac{1}{\epsilon'}) = \grandO(\frac{n}{\epsilon'})$, so Observation~\ref{obs:dp} leads to the desired complexity of $\grandO(n^3 \NU^2 (\frac{n}{\epsilon'})^{\NU})$
to get a ratio $1+2n\epsilon'$. Finally, given a target ratio $1+\epsilon$, we set $\epsilon' = \frac{\epsilon}{2n}$, 
obtaining the claimed complexity.
\end{proof}

\revision{
\section{Conclusion}
In this paper, we proved that for a certain type of problem $\Pi \in \PBFAMILY$, we can transfer the approximability of $\Pi$ to its robust variant \ROBUST{$\Pi$} (to the price of a constant multiplicative factor),
and we provided an \FPTAS for the robust shortest path problem. Natural questions for future research would be to improve some of the multiplicative constants we obtained for specific structures and to generalize the \FPTAS to a larger class of problems.
}

\nocite{*}
\bibliographystyle{abbrvnat}
\bibliography{locational_uncertainty}
\label{sec:biblio}

\end{document}